\documentclass{article}
\usepackage[margin=1in]{geometry}

\usepackage[T1]{fontenc}
\usepackage{amsfonts,amsmath,amsthm,amssymb,mathtools}  %

\mathtoolsset{centercolon}
\usepackage{xfrac,nicefrac}
\usepackage{mathdots}
\usepackage{mleftright}  %

\usepackage{xspace}
\xspaceaddexceptions{]\}}  %
\usepackage{regexpatch}

\usepackage{bm,bbm,dsfont}  %
\usepackage{caption}
\usepackage[normalem]{ulem}
\usepackage{enumitem}

\usepackage{graphicx}
\usepackage{float}
\usepackage{subcaption}  %
\usepackage{tcolorbox}
\usepackage{tikz}
\usetikzlibrary{decorations.pathreplacing}
\usetikzlibrary{calc}
\usetikzlibrary{positioning}
\usetikzlibrary{arrows.meta}

\usepackage[linesnumbered,boxed,ruled,vlined]{algorithm2e}
\usepackage{algpseudocode}

\usepackage{thmtools,thm-restate}
\theoremstyle{plain}
\newtheorem{prob}{Problem}  %

\newtheorem{theorem}{Theorem}[section]  %
\newtheorem{lemma}[theorem]{Lemma}

\newtheorem{cor}[theorem]{Corollary}

\theoremstyle{definition}  %
\newtheorem{definition}[theorem]{Definition}

\usepackage[colorlinks,citecolor=blue,linkcolor=blue,urlcolor=red]{hyperref}
\usepackage[capitalise]{cleveref}
\crefname{algocf}{Algorithm}{Algorithms}
\Crefname{algocf}{Algorithm}{Algorithms}
\crefname{prob}{Problem}{Problems}

\DeclarePairedDelimiter{\ceil}{\lceil}{\rceil}
\DeclarePairedDelimiter{\floor}{\lfloor}{\rfloor}

\DeclarePairedDelimiter{\bk}{(}{)}
\DeclarePairedDelimiter{\Bk}{[}{]}
\DeclarePairedDelimiter{\BK}{\{}{\}}

\DeclarePairedDelimiter{\abs}{\lvert}{\rvert}

\DeclareMathOperator{\poly}{poly}

\newcommand{\eqdef}{\eqqcolon}
\newcommand{\defeq}{\coloneqq}
\newcommand{\eps}{\varepsilon}

\renewcommand{\l}{\ell}

\renewcommand{\epsilon}{\eps}
\newcommand{\Patrascu}{\textup{P{\v{a}}tra{\c{s}}cu}\xspace}

\newcommand{\numberthis}{\addtocounter{equation}{1}\tag{\theequation}}

\newcommand{\xlow}[1][i]{x_{#1}^{(\textup{low})}}
\newcommand{\xmidhigh}[1][i]{x_{#1}^{(\textup{mid-high})}}
\newcommand{\deltamid}[1][i]{\delta_{#1}^{(\textup{mid})}}
\newcommand{\deltahigh}[1][i]{\delta_{#1}^{(\textup{high})}}
\newcommand{\height}{h}

\newcommand{\rank}{\textup{\textsc{Rank}}\xspace}
\newcommand{\select}{\textup{\textsc{Select}}\xspace}
\newcommand{\partialsum}{\textup{\textsc{PartialSum}}\xspace}

\renewcommand{\vec}[1]{\bm{#1}}

\newcommand{\Lthrd}{L_{\textup{thrd}}}
\newcommand{\vhigh}{v_{\textup{high}}}

\newcommand{\vlow}{v_{\textup{low}}}

\newcommand{\defn}[1]{\emph{\textbf{#1}}}

\usepackage{regexpatch}
\makeatletter
\xpatchcmd\thmt@restatable{%
\csname #2\@xa\endcsname\ifx\@nx#1\@nx\else[{#1}]\fi
}{%
\ifthmt@thisistheone
\csname #2\@xa\endcsname\ifx\@nx#1\@nx\else[{#1}]\fi
\else
\csname #2\@xa\endcsname[{Restated}]
\fi}{}{}
\makeatother

\usepackage{array}

\title{Optimal Static Fully Indexable Dictionaries}
\author{
Jingxun Liang\thanks{Carnegie Mellon University. \texttt{jingxunl@andrew.cmu.edu}.}
\and
Renfei Zhou\thanks{Carnegie Mellon University. Partially supported by the MongoDB PhD Fellowship. \texttt{renfeiz@andrew.cmu.edu}.}
}
\date{}

\begin{document}

\maketitle

\begin{abstract}
  Fully indexable dictionaries (FID) store sets of integer keys while supporting rank/select queries. They serve as basic building blocks in many succinct data structures. Despite the great importance of FIDs, no known FID is \emph{succinct} with efficient query time when the universe size $U$ is a large polynomial in the number of keys $n$, which is the conventional parameter regime for dictionary problems. In this paper, we design an FID that uses $\log \binom{U}{n} + \frac{n}{(\log U / t)^{\Omega(t)}}$ bits of space, and answers rank/select queries in $O(t + \log \log n)$ time in the worst case, for any parameter $1 \le t \le \log n / \log \log n$, provided $U = n^{1 + \Theta(1)}$. This time-space trade-off matches known lower bounds for FIDs \cite{patrascu2006timespace, patrascu2010cellprobe, viola2023new} when $t \le \log^{0.99} n$.
  
  Our techniques also lead to efficient succinct data structures for the fundamental problem of maintaining $n$ integers each of $\l = \Theta(\log n)$ bits and supporting partial-sum queries, with a trade-off between $O(t)$ query time and $n\l + n / (\log n / t)^{\Omega(t)}$ bits of space. Prior to this work, no known data structure for the partial-sum problem achieves constant query time with $n \l + o(n)$ bits of space usage.
\end{abstract}

\section{Introduction}
\label{sec:introduction}
A \defn{fully indexable dictionary} (a.k.a.~\defn{rank/select dictionary}; \defn{FID} for short) is a fundamental data structure, which stores a set $S$ of $n$ keys from a universe $[U]$, supporting \rank and \select queries:
\begin{itemize}
\item $\rank(x)$: return the number of keys in $S$ that are smaller than or equal to $x$.
\item $\select(i)$: return the $i$-th smallest key in $S$.
\end{itemize}
This paper focuses on \defn{static} FIDs, where the key set is given at the beginning and does not change over time. There is also the dynamic case, allowing updates to the key set via insertions and deletions, which is not considered in this paper.

FIDs are powerful data structures with numerous applications. First, if we choose to think of $S$ as an indicator vector in $\BK{0, 1}^{U}$, then the problem becomes storing a bit string containing $n$ ones and $U - n$ zeros, while supporting prefix-sum queries and queries for the position of the $i$-th one. This so-called \defn{rank/select} problem serves as a subroutine in many space-efficient data structures. Second, given \rank and \select, one can also support \defn{predecessor search} queries (i.e.,~find the largest element in $S$ that does not exceed $x$) by calling $\select(\rank(x))$---this, too, serves as a common data-structural subroutine. The result is that FIDs have many applications in data structures for strings \cite{clark1996efficient, ferragina2005indexing, grossi2005compressed, hon2009breaking, munro2001space, navarro2007compressed, grossi2003highorder, ferragina2008searching}, trees \cite{ferragina2005structuring, geary2006succinct, raman2007succinct, munro2001succinct}, parentheses sequences \cite{geary2006simple, munro2001succinct}, multisets \cite{raman2007succinct}, permutations \cite{munro2012succinct}, variations of dictionaries \cite{buhrman2002are, blandford2008compact}, etc.

The most fundamental question regarding FIDs is to determine their best possible time-space trade-off. If an FID storing $n$ keys from the universe $[U]$ uses $\log \binom{U}{n} + R$ bits of space, we say it incurs $R$ bits of \defn{redundancy}, where the first term $\log \binom{U}{n}$ is referred to as the \emph{information-theoretic optimum}.
Conventionally, an FID is said to be \defn{compact} if the redundancy $R = O\bk[\big]{\log \binom{U}{n}}$, and is said to be \defn{succinct} if $R = o\bk[\big]{\log \binom{U}{n}}$.
The \defn{time-space trade-off} of an FID is the relationship between the redundancy $R$ and the \emph{worst-case query time} under a RAM with word-size $\Theta(\log U)$.

\paragraph{Towards the optimal trade-off.}

There is vast literature on space-efficient FIDs. The earliest works \cite{jacobson1988succinct,munro1996tables,clark1996compact} on FIDs focused on the setting where one needs to store $S$'s indicator vector using $U$ bits in the plain format, while using up to $R$ additional bits to store auxiliary information. Compared to the information-theoretic optimum $\log \binom{U}{n}$, their approaches are succinct only if $U = (2 \pm o(1)) n$.
In 2002, Raman, Raman, and Satti \cite{raman2007succinct} constructed an FID with redundancy $O(U \log \log U / \log U)$ and constant query time. Later, \Patrascu \cite{patrascu2008succincter} reduced the redundancy to $U / (\log U / t)^{\Omega(t)}+ O(U^{3/4})$ bits with worst-case query time $O(t)$.

Some applications of FIDs demand to store dense sets where the universe size $U = O(n)$; in this case, \Patrascu's FID is already succinct and achieves an ideal trade-off between $O(t)$ query time and $n / (\log n / t)^{\Omega(t)}$ bits of redundancy, which is provably optimal when $t \le \log^{0.99} n$ \cite{patrascu2010cellprobe, viola2023new}. However, when $U = n^{1 + \Theta(1)}$, which is the conventional parameter regime for dictionary problems, the redundancy of \Patrascu's FID will be polynomially larger than the information-theoretic optimum $\log \binom{U}{n} = \Theta(n \log n)$ even for large running times, e.g., $t = \log^{0.99} n$, and thus is not \emph{succinct} or \emph{compact}.
Gupta et al.~\cite{gupta2007compressed} further improved the redundancy to $O(n \log \log n)$ bits, thus being succinct, at the cost of slower queries of $O(\log^2 \log n)$ time. It is currently not known how to construct succinct FIDs with query time $o(\log^2 \log n)$.

On the lower-bound side, \Patrascu and Thorup \cite{patrascu2006timespace} showed an $\Omega(\log \log n)$ time lower bound for \emph{predecessor} queries, assuming that $U = n^{1 + \Theta(1)}$ and that the data structure is compact. Their lower bound also applies to FIDs by a reduction. Another lower bound proven by \Patrascu and Viola \cite{patrascu2010cellprobe, viola2023new} shows that, if an FID takes $O(t)$ worst-case time per query, it must incur $n / (\log n)^{O(t)}$ bits of redundancy.
The best-known lower bound for FIDs is a simple combination of these two independent lower bounds: When the redundancy equals $n / (\log n)^{\Omega(t)}$, the worst-case query time must be at least $\Omega(\max\BK{\log \log n, \, t})$, provided $U = n^{1 + \Theta(1)}$. There remains a huge gap between the lower and upper bounds.

In summary, it has remained one of the basic open questions in the field whether one can hope to construct an FID that is both succinct and that supports queries in the optimal time of $O(\log \log n)$. Moreover, although there are well-established lower bounds for the time-space trade-off of a static FID, it remains open on the upper-bound side to obtain tight bounds \emph{at any point along the curve}. These are the problems that we seek to resolve in the current paper. 

\paragraph{This paper: Tight upper bounds for FIDs.}

In this paper, we construct a succinct FID as shown in the following theorem.

\begin{restatable}{theorem}{thmMainFID}
  \label{thm:main-fid}
  For any parameters $n, U, t$ with $U = n^{1 + \Theta(1)}$ and $t \le \log n/\log\log n$, there is a static fully indexable dictionary with query time $O(t + \log \log n)$ and redundancy $R = n/(\log n/t)^{\Omega(t)}$, in word RAM with word size $w = \Theta(\log n)$.
\end{restatable}

Thus, it is possible to achieve succinctness while offering an optimal time bound of $t = O(\log \log n)$---in this case, the redundancy of our construction is $n / (\log n)^{\Omega(\log \log n)} \ll n / \poly \log n$ bits. More generally, the time-space trade-off offered by the above theorem is provably optimal \emph{for all} $t \le \log^{0.99} n$, as, in every parameter regime, it matches one of the two known lower bounds. Somewhat surprisingly, this means that the maximum of the two completely independent lower bounds forms a single tight lower bound for FIDs.

\paragraph{High-level technical approach.} To understand the high-level approach taken by our data structure, let us consider the lower bound in \cite{patrascu2006timespace}, which points out a distribution of hard instances such that any data structure with near-linear space $O(n \poly \log n)$ needs to spend $\Omega(\log \log n)$ time for each query on these inputs.

A critical insight in the current paper is that, although it is difficult to improve the query time for these hard instances, they are well-structured so that they only occupy a small fraction of all possible inputs. Therefore, in principle, the space we actually need to store these hard instances is small.

Moreover, if we restrict only to these hard instances, then we can afford to use a data structure that is space-efficient compared to the optimal bound for FIDs but that is \emph{space-inefficient} compared to the information-theoretic optimum for hard instances. So we can handle hard instances by creating data structures that are morally space inefficient (i.e., not optimal for these hard instances), but that are nonetheless space efficient in the context of the overall FID problem.

On the flip side, when an input is not hard, we can think of it morally as being like a ``random'' input. The entropy of random inputs is large, so we cannot waste much space. Fortunately, random instances can benefit from the same high-level techniques that have already been developed in past work \cite{patrascu2008succincter}, allowing for fast queries with good space efficiency. 

Of course, ``hard instances'' and ``random instances'' are really just two extremes in a large design space of inputs. What is remarkable is that, nonetheless, it is possible to combine the high-level approaches described above in order to construct a single data structure that achieves the optimal time-space trade-off for all inputs. 

\paragraph{Other implications of our techniques.}

As mentioned earlier, FIDs can naturally support predecessor search queries, so our result also improves the predecessor data structures:

\begin{cor}
    \label{cor:predecessor}
    For any parameters $n, U$ with $U = n^{1+\Theta(1)}$, there is a data structure storing a set $S$ of $n$ keys from a universe $[U]$ that supports predecessor searches in $S$, with the same time-space trade-off as in \cref{thm:main-fid}.
\end{cor}

Prior to this result, it remained open whether one could achieve $O(\log \log n)$ query time while also offering $o(n)$ bits of redundancy \cite{pibiri2017dynamic}.

\smallskip

By adjusting our data structure for FIDs, we show the following time-space trade-off for so-called \defn{select dictionaries}, which are dictionaries that support \select but not \rank:

\begin{restatable}{theorem}{thmSelect}
  \label{thm:select}
  For any parameters $n, U, t$ with $U = n^{1 + \Theta(1)}$ and $t \le \log n / \log \log n$, there is a static dictionary that answers \textup{\select} queries within $O(t)$ time and has redundancy $R = n / (\log n / t)^{\Omega(t)}$, in word RAM with word size $w = \Theta(\log n)$.
\end{restatable}

Here, again, the bounds that we achieve are provably optimal, as they match a lower bound previously proven by \cite{patrascu2010cellprobe}. This represents a significant improvement over the previous state of the art \cite{raman2007succinct}, which was unable to achieve space bounds any better than $O(n / \sqrt{\log n})$.

\smallskip

Another application of our techniques is to the basic question of storing $n$ $\Theta(\log n)$-bit integers and answering partial-sum queries.

\begin{restatable}{theorem}{thmPartialSum}
  \label{thm:partialsum}
  For any parameters $n, \l, t$ with $\l = \Theta(\log n)$ and $t \le \log n / \log \log n$, there is a data structure storing a sequence of $n$ $\l$-bit integers $a_1, \ldots, a_n$, which uses $n\l + n / (\log n / t)^{\Omega(t)}$ bits of space and can answer the partial-sum queries $\sum_{j=1}^{i} a_j$ for any given $i$ within $O(t)$ time.
\end{restatable}

For this partial-sum problem, the previous state of the art with $O(1)$ query times incurred $O(n)$ bits of redundancy \cite{raman2001succinct}. Whether or not this could be reduced to $o(n)$ remained an open question. Our bound reduces it all the way to $n / \poly \log n$ bits for a polylogarithm of our choice.

\subsection{Related Works}

Grossi et al.~\cite{grossi2009more} studied the FID problem for polynomial universe sizes with $O(1)$ query time. In this setting, achieving a compact space usage of $O(n \log n)$ bits is impossible due to the lower bound established by \cite{patrascu2006timespace}. They showed a trade-off between $O(\eps^{-1})$ query time and $O(n^{1 + \eps})$ bits of space.

In the special case where the universe size $U = 2n$, Yu~\cite{yu2019optimal} presented a data structure that supports $\rank$ queries in $O(t)$ time, but not $\select$ queries. This data structure incurs only $n / (\log n)^{\Omega(t)} + n^{1 - \Omega(1)}$ bits of redundancy, outperforming \Patrascu's data structure when $t = (\log n)^{1 - o(1)}$ is large, and matching the lower bound of \cite{patrascu2006timespace} for all $t$.

A closely related setting is the \defn{dynamic FID} problem, where, in addition to answering $\rank$ and $\select$ queries, the data structure also needs to support fast insertions and deletions. When the universe size $U$ is linear in the number of keys $n$, Li et al.~\cite{li2023dynamic} achieved $O(n / 2^{\log^{0.199} n})$ bits of redundancy with the optimal operational time of $O(\log n / \log \log n)$. For polynomial universe sizes $U = n^{1 + \Theta(1)}$, it remains open whether a succinct data structure with $O(\log n / \log \log n)$ time per operation can be constructed.

Predecessor data structures are also closely related to FIDs and have been extensively studied across various parameter regimes and settings due to their importance. For a comprehensive overview, see the survey by Navarro and Rojas-Ledesma~\cite{navarro2021predecessor}.

Another variant of dictionary is the \defn{unordered dictionary} problem, where the data structure only needs to answer membership queries---whether a given key $x$ is in the current key set. A series of works has focused on both static and dynamic unordered dictionaries~\cite{yu2020nearly,hu2025optimal,raman2003succinct,li2023tight,bender2022optimal,li2024dynamic,bender2024modern}, leading to a static unordered dictionary with constant query time and $n^{\eps}$ bits of redundancy~\cite{hu2025optimal}, as well as tight upper and lower bounds for the time-space trade-off of dynamic unordered dictionaries for polynomial universe sizes~\cite{li2023tight,bender2022optimal,li2024dynamic}.

Raman, Raman, and Satti~\cite{raman2007succinct} studied the \defn{indexable dictionary (ID)} problem, which is similar to FIDs, but $\rank(x)$ only returns the rank of $x$ when $x$ is present in the current key set; otherwise, it returns ``not exist.'' They constructed IDs with $O(n / \sqrt{\log n})$ bits of redundancy and $O(1)$ query time. Note that the ID setting is easier than FIDs, allowing data structures to achieve both succinctness and constant-time queries.

A further relaxation of ID is \defn{monotone minimal perfect hashing (MMPH)}, where the data structure only needs to support $\rank$ queries for elements in the key set---if $x$ is not in the key set, $\rank(x)$ can return an arbitrary answer. ($\select$ queries are not required.) The key point of this relaxation is that MMPH data structures use asymptotically less space: Belazzougui, Boldi, Pagh, and Vigna~\cite{belazzougui2009monotone} constructed a data structure that uses only $O(n \log \log \log U)$ bits, while encoding the key set itself requires $\log \binom{U}{n} = O(n \log (U/n))$ bits. This bound is shown to be optimal by Assadi, Farach-Colton, and Kuszmaul~\cite{assadi2024tight} (see also \cite{kosolobov2024simplified}).

There is also a line of research on rank/select problems over arbitrary alphabets~\cite{grossi2003highorder,golynski2006rank,ferragina2007compressed,golynski2008redundancya,grossi2010optimal,barbay2014efficient,belazzougui2015optimal}. Given a sequence in $[\sigma]^n$, the select query asks for the $k$-th occurrence of a given symbol $s \in [\sigma]$; rank queries are defined similarly. These problems generalize FIDs, which correspond to the special case $\sigma = 2$. For arbitrary $\sigma$, the optimal time/redundancy trade-off is still not fully understood~\cite{belazzougui2015optimal}.

\section{Preliminaries}
\paragraph*{Augmented B-trees.}
We will use the augmented B-trees (aB-trees for short) from \cite{patrascu2008succincter} as a subroutine. 

Let $B$ and $m$ be parameters such that $m$ is a power of $B$, and let $A[1\ldots m]$ be an array of elements in the alphabet $\Sigma$. An aB-tree of branching factor $B$ and size $m$ is a full $B$-ary tree over $m$ leaves, which correspond to the entries of the array $A[1\ldots m]$. Additionally:
\begin{itemize}
  \item Each node of the aB-tree is augmented by a label from the \emph{label alphabet} $\Phi$. The label of a leaf node is determined by the corresponding entry $A[i]$ in the array, and the label of an internal node is determined by the label sequence of its children. Formally, there is a transition function $\mathcal{A}$ determining the label $\phi_u$ of each the internal node $u$: $\phi_u = \mathcal{A}(\vec{\phi})$, where $\vec{\phi} \defeq \bk{\phi_1, \phi_2, \ldots, \phi_B}$ is the label sequence of $u$'s $B$ children.
  \item There is a recursive query algorithm, which starts by examining the label of the root, and then recursively traverses down a path from the root to some leaf of the tree. At each step, the algorithm examines the label of the current node and the labels of its children to determine which of the children to recurse on. After reaching a leaf, the algorithm outputs the answer to the query based on all the examined labels. Furthermore, this algorithm is restricted to spend only $O(1)$ time at each examined node, ensuring that the query time remains at most $O(\log_B m)$.
\end{itemize}

Beyond \cite{patrascu2008succincter}, we also consider incomplete aB-trees. Let $m$ be any integer, not necessarily a power of $B$, and let $t$ be an integer such that $B^t \ge m$. An \emph{incomplete aB-tree} with branching factor $B$ and size $m$ is derived from a (full) aB-tree over $B^t$ leaves with the same branching factor $B$, by retaining only the first $m$ leaves while removing other $B^t - m$ leaves, and (repeatedly) deleting all internal nodes without any child. In such an incomplete aB-tree, the transition function $\mathcal{A}$ of labels still follows the form $\phi_u = \mathcal{A}(\vec{\phi})$, but the input label sequence $\vec{\phi} = \bk{\phi_1, \ldots, \phi_{\l}}$ may have a length $\l$ less than $B$, as $u$ might possess fewer than $B$ children.

Let $\mathcal{N}(m, \phi)$ be the number of instances of $A[1\ldots m]$ such that the aB-tree of branching factor $B$ over it will have root label $\phi$. According to Theorem 8 of \cite{patrascu2008succincter}, when $m$ is a power of $B$, we can compress the aB-tree with size $m$ and root label $\phi$ to within $\log \mathcal{N}(m, \phi) + 2$ bits. Their proof directly works for incomplete aB-trees ($m \le B^t$) as well.

\begin{lemma}[Natural generalization of {\cite[Theorem 8]{patrascu2008succincter}}]
  \label{lm:succincter}
  Let $m, B, t$ be parameters with $m \le B^t$ and $B = O\bk{\frac{w}{\log\bk{m + \abs{\Phi}}}}$. Suppose there is an (incomplete) aB-tree with branching factor $B$, size $m$, and root label $\phi$, then we can compress this aB-tree to within $\log \mathcal{N}(m, \phi) + 2$ bits with query time $O(t)$, in the word RAM model with word size $w$. The data structure uses a lookup table of $O\bk{B^{2t} \abs{\Sigma} + B^{3t} \abs{\Phi}^{2B}} $ words, which only depends on $B$ and $t$.
\end{lemma}

Notice that the lookup table only depends on $B$ and $t$, but not $m$. The original construction in \cite{patrascu2008succincter}, when applied on incomplete aB-trees, uses a lookup table of $O\bk{B(|\Sigma| + |\Phi|^{B+1} + B |\Phi|^B)}$ words, which depends on $B$, $t$, and $m$. To avoid the dependency on $m$, we simply concatenate $B^t$ such lookup tables for all values of $m$ together, with at most $O(B^{t+1}\cdot \bk{\abs{\Sigma} + \abs{\Phi}^{B+1} + B\abs{\Phi}^B}) \le O(B^{2t} \abs{\Sigma} + B^{3t}\abs{\Phi}^{2B})$ words, and let aB-trees with different sizes $m$ use different parts of the (concatenated) lookup table. Later when we apply this lemma in our data structure, we will let multiple instances with the same $B$ and $t$ (but with different $m$) share a single lookup table.

\paragraph*{Predecessor data structures.}
We will use the following extension of predecessor data structure as a subroutine.

\begin{prob}[Predecessor with associated values]
  Storing a set $S \subset [U]$ of $n$ keys,\footnote{We use notation $[n]$ to denote the set $\BK{0,1,\ldots,n-1}$ for any non-negative integer $n$, and use $[a, b]$ to denote the set of integers $\BK{a, a+1, \dots, b}$ when there is no ambiguity.} where each key is associated with a value in $[V]$, supporting predecessor and successor queries:
  \begin{itemize}
  \item $\textup{\textsc{Predecessor}}(x)$: Return the largest element $y \in S$ such that $y \le x$, and the associated value of $y$.
  \item $\textup{\textsc{Successor}}(x)$: Return the smallest element $y \in S$ such that $y > x$, and the associated value of $y$.
  \end{itemize}   
\end{prob}
By studies on dictionaries \cite{fredman1984storing} and predecessor data structures \cite{patrascu2006timespace}, there is a compact construction for predecessor data structures with associated values:
\begin{lemma}
  \label{lm:predecessor}
  For any $U \ge n$, there are predecessor data structures with associated values using $O\bk[\big]{n\log U + n \log V}$ bits, with query time $O(\log \log (U/n))$ in the worst case.
\end{lemma}
\begin{proof}[Proof Sketch] We maintain the following data structures simultaneously:
  \begin{itemize}
  \item A predecessor data structure for the set $S \subset [U]$ by \cite{patrascu2006timespace} with space $O(n \log U)$ and query time $O(\log \log (U/n))$.
  \item A successor data structure similar to the predecessor data structure.
  \item A ``perfect hashing'' by \cite{fredman1984storing,cormen2022introduction} to store the key set $S$ with their associated values, with a space complexity of $O(n (\log U + \log V))$ bits and a worst-case query time $O(1)$.\qedhere
  \end{itemize}
\end{proof}

Throughout this paper, we use the terminology ``predecessor data structure'' for short to refer to the data structure defined by \cref{lm:predecessor}.

\section{Basic Data Structure for FIDs}
\label{sec:simple}

In this section, we will construct a basic data structure for FIDs with a slightly worse time and space guarantee than the requirement of \cref{thmt@@thmMainFID}. It will serve as a subroutine in the final data structure in \cref{sec:advanced}: The final data structure is based on the algorithm framework in this section, and by replacing a subroutine with the result in this section (\cref{thm:simple_FID}) in a non-recursive way, it achieves the ideal time-space trade-off. Formally, we will prove the following theorem:

\begin{theorem}[Weak version of \cref{thmt@@thmMainFID}]
  \label{thm:simple_FID}
  For any parameters $U, n, t$ with $U = n^{1 + \Omega(1)} $, $t \le \log U / \log \log U$ and a constant $\eps > 0$, there is a static fully indexable dictionary with query time $O(t\log \log U)$ and redundancy $R = \max\BK{n/(\log U/t)^{\Omega(t)}, O(\log U)}$, in the word-RAM model with word size $w = \Theta(\log U)$. The data structure uses a lookup table of $O(U^{10\eps})$ words that only depend on $U$ and $t$ which could be shared among multiple instances of fully indexable dictionaries.
\end{theorem}

Although in most applications of FIDs we care about polynomially-sized universes $U = n^{1 + \Theta(1)}$, here we also consider the parameter regime where $n = U^{o(1)}$ is significantly smaller than $U$. The reason is that, later in \cref{sec:advanced}, we will use \cref{thm:simple_FID} to maintain short subsequences of keys.

In the remainder of this section, we construct this basic data structure to prove \cref{thm:simple_FID}. \cref{tab:notations} lists the main parameters and notations we will introduce in this section.

\begin{table}[h!]
    \centering
    \caption{Table of Notations}\label{tab:notations}
    \begin{tabular}{|>{\centering\arraybackslash}m{3cm}|>{\arraybackslash}m{12cm}|}
        \hline
        \textbf{Notation} & \multicolumn{1}{c|}{\textbf{Explanation}} \\ \hline
        $t$ &
        A parameter indicating that our algorithm's time constraint is $O(t \log \log U)$.
        \\ \hline
        $R$ &
        $R \defeq \max\BK{n / (\log U / t)^{\Omega(t)}, \, O(\log U)}$ is the desired redundancy of our FID.
        \\ \hline
        $B$ &
        The branching factor of aB-trees in the mid parts. $B \log B = (\epsilon \log U) / t$. When $t \le \log^{0.99} U$, there is
        $B = \log^{\Theta(1)} U$.
        \\ \hline
        $B^t$ &
        The number of keys in each block.
        \\ \hline
        $h$ &
        $h \defeq t \log B$. Each mid part consists of $2h$ bits.
        \\ \hline
        $b$ &
        $b \defeq \log(U/n) - h$. Each low part consists of $b$ bits.
        \\ \hline
        $\xlow$,~$\deltamid$,~$\deltahigh$ &
        The value in the low, mid, and high parts. See \cref{fig:partition_combined}.
        \\ \hline
        $\delta_i$ &
        The value in the mid and high parts together. It equals $\deltamid + 2^{2h} \cdot \deltahigh$.
        \\ \hline
        $\delta_{\le i}$ &
        $\delta_{\le i} \defeq \sum_{j=1}^i \delta_j$ is the prefix sum of $\BK{\delta_i}$.
        \\ \hline
        $\Delta$ &
        $\Delta = \delta_{\le B^t}$ is the summation of all $\delta_i$ within a block.
        \\ \hline
    \end{tabular}
\end{table}

\paragraph*{Partitioning into blocks.} Let $S = \BK{x_1, x_2, \ldots, x_n}$ be the set of keys we need to store, where $x_1 < x_2< \cdots < x_n$. The first step of the construction is to divide the sequence $\bk{x_1, x_2, \ldots, x_n}$ into small blocks, and to store some inter-block data structure that reduces the entire FID problem to small-scale FID problems within each block.

Let $B$ be a parameter where $B \log B = \frac{\eps \log U}{t}$, and we break the whole sequence into blocks of size $B^t$. As the sequence $\bk{x_1, x_2, \ldots, x_n}$ is monotonically increasing, the partition into blocks could be viewed as partitioning the possible range of the keys $[U]$ into $n/B^t$ disjoint intervals, where the $k$-th block corresponds to the interval $(x_{(k-1)B^t}, x_{k\cdot B^t}]$.\footnote{We let $x_0 = 0$ for convenience. When $n$ is not divisible by $B^t$, especially when $n < B^t$, the size of the last block will be smaller than $B^t$. Fortunately, the same construction below works for any block size $m \le B^t$, and we only illustrate our construction for block size $B^t$.}

Our inter-block data structure consists of the following two parts:
\begin{itemize}
\item The sequence of endpoints of each interval, i.e., $\bk{x_{B^t}, x_{2\cdot B^t}, \ldots, x_{n}}$.
\item A predecessor structure for the sequence of endpoints $\bk{x_{B^t}, x_{2\cdot B^t}, \ldots, x_{n}}$, where each entry $x_{k\cdot B^t}$ is associated with value $k$.
\end{itemize}
Given the above auxiliary information, we view each block $k$ as an FID with $B^t$ keys from a universe of size $x_{k B^t} - x_{(k - 1)B^t}$. When we perform a rank query $\rank(x)$, the second part above can help us locate the block $k$ containing the queried element $x$ (i.e., $x \in (x_{(k-1)B^t},\, x_{k B^t}]$), transforming the original query into a \rank query in the $k$-th block, within $O(\log \log U)$ time (see \cref{lm:predecessor}). When we perform a select query $\select(i)$, it becomes a \select query within block $k = \ceil{i / B^t}$. The remaining question is to answer \rank/\select queries within each block.

\paragraph*{Storing difference sequences within blocks.}
Throughout the remainder of this section, we mainly focus on the FID problem within a single block $k$. Letting $s = (k - 1) B^t$, the intra-block FID problem requires us to maintain the sequence of keys $(x_{s+1}, \ldots, x_{s + B^t})$ in the universe $(x_s, x_{s + B^t}]$. 
We cut the binary representation of each key $x_{s + i}$ into two parts, as shown in \cref{fig:partition_low}:
\begin{itemize}
\item Letting $h$ be a parameter to be determined and $b \defeq \log \frac{U}{n} - h$, we define the \defn{low part} as the $b$ least significant bits in the binary representation of each key $x_{s + i}$, denoted by integers $\xlow \in [2^b]$. We will directly store these integers in the data structure.
\item The $\log U - b$ remaining (more significant) bits are referred to as the \defn{mid-high part}, denoted by integers $\xmidhigh$. For these integers, we aim to store the \defn{difference sequence}, i.e., the differences $\delta_i \defeq \xmidhigh[i] - \xmidhigh[i-1]$ between adjacent pairs of elements.
\end{itemize}
Clearly, we have
\[ x_{s+i} = \xlow + \xmidhigh \cdot 2^b \qquad \bk{0 \le i \le B^t}. \]
For short, we let $\delta_{\le i}$ denote the partial sum $\sum_{j = 1}^i \delta_j$, so that $\xmidhigh[i] = \delta_{\le i} + \xmidhigh[0]$.
We further define $\Delta \defeq \delta_{\le B^t}$ as the sum of the difference sequence $(\delta_1, \ldots, \delta_{B^t})$:
\begin{align*}
  \Delta \defeq \sum_{i=1}^{B^t} \delta_i = \xmidhigh[B^t] - \xmidhigh[0] = \floor*{\frac{x_{s+B^t}}{2^b}} - \floor*{\frac{x_{s}}{2^b}}.
\end{align*}
We denote by $\Delta_{k'}$ the value of $\Delta$ in the $k'$-th block. Clearly,
\begin{align*}
  \label{ineq:sum_of_delta_k}
  \sum_{k'=1}^{n/B^t} \Delta_{k'} = \floor*{\frac{x_{n}}{2^b}} - \floor*{\frac{x_0}{2^b}} \le \frac{U}{2^b} = n\cdot 2^h.\numberthis
\end{align*}

With these notations, the FID problem within a single block can be restated as follows:
\begin{prob}[FID within a block]
  \label{prob:FID_within_block}
  Let $\Delta$ be a parameter stored outside. We need to store two sequences $\bk{\xlow[1], \xlow[2], \ldots, \xlow[B^t]}$ and $\bk{\delta_1, \delta_2, \ldots, \delta_{B^t}}$, such that $\xlow \in [2^b]$ and $\sum_{i=1}^{B^t} \delta_i = \Delta$, with the property that $\bk[\big]{\xlow + \delta_{\le i} \cdot 2^b}_{i = 1}^{B^t}$ forms a strictly increasing sequence, supporting the following queries:
  \begin{itemize}
  \item $\partialsum(i)$: Return $\delta_{\le i} \cdot 2^b + \xlow$. It corresponds to the \select queries in the original FID problem.
  \item $\rank(x)$: Return the largest $i$ such that $\delta_{\le i} \cdot 2^b + \xlow \le x$.
  \end{itemize}
\end{prob}

As introduced above, once we have a solution to \cref{prob:FID_within_block} with $O(T)$ time for any $T$, with the help of the inter-block data structure, we may answer the \rank/\select queries to the original key sequence in $O(T + \log \log U)$ time.

\begin{figure}[ht]
  \centering
  \begin{subfigure}[c]{6.5cm}
    \begin{subfigure}[c][4cm][c]{6.5cm}
      \vspace{0.1cm}
      \hspace{0.4cm}
      \includegraphics[width = 7.4118cm]{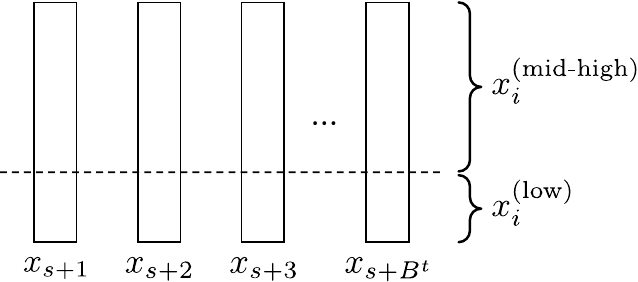}
    \end{subfigure}
    \caption{Separating low parts from mid-high parts.}
    \label{fig:partition_low}
    \centering
  \end{subfigure}
  \hspace{1cm}
  \hspace{0.5cm}  %
  \begin{subfigure}[c]{6cm}
    \begin{subfigure}[c][4cm][c]{6cm}
      \includegraphics[width = 7cm]{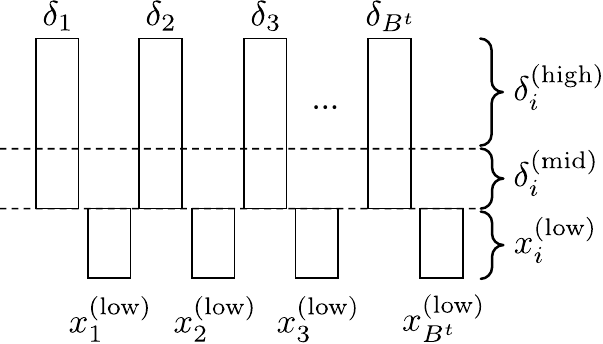}
    \end{subfigure}
    \caption{Separating mid parts and high parts.}
    \label{fig:3partition}
    \centering
  \end{subfigure}
  \hspace{1cm}
  \caption{Partitioning keys into three parts. We (a) divide the binary representations of the keys into low parts and mid-high parts, and further (b) take the difference sequence $\bk{\delta_i}_{1 \le i \le B^t}$ of the mid-high parts and divide them into mid parts and high parts. There are $b \defeq \log (U/n) - h$ bits in the low part, $2h$ bits in the mid part, and $\log n - h$ bits in the high part.}
  \label{fig:partition_combined}
\end{figure}

\paragraph*{The three-part partition.}
For some technical reasons, we further divide each $\delta_i$ (i.e., the difference sequence of the mid-high part) into two smaller parts as follows. We will use different approaches to organize these parts in \cref{sec:simple/ds}.
\begin{itemize}
\item Let $\deltamid$ be the integer formed by the $2h$ least significant bits of $\delta_i$, which we call the \defn{mid part}.
\item Let $\deltahigh$ be the integer formed by the remaining $\log U - b - 2h = \log n - h$ bits, which we call the \defn{high part}.
\end{itemize}
Formally,
\[\delta_i = \deltamid + \deltahigh \cdot 2^{2\height} \qquad (i \in [1, B^t]).\]

By now, we have cut the sequence of keys $(x_{s+1}, \ldots, x_{s+B^t})$ in a single block into three parts: the high part $\bk{\deltahigh[1], \ldots, \deltahigh[B^t]}$, the mid part $\bk{\deltamid[1], \ldots, \deltamid[B^t]}$, and the low part $\bk{\xlow[1], \ldots, \xlow[B^t]}$, as shown in \cref{fig:3partition}. The following subsection \ref{sec:simple/ds} will design (almost) separate data structures to store these three parts.

\subsection{Data Structure within Each Block}
\label{sec:simple/ds}
In this subsection, we will design the intra-block data structure with time complexity $O(t \log \log U)$ of three parts.
At a very high level, these three parts of the data structure will interact as follows when performing \partialsum and \rank queries:
\begin{itemize}
  \item During the query $\partialsum(i)$, we will compute $\deltahigh[\le i]$, $\deltamid[\le i]$ and $\xlow$ from the high part, mid part and the low part of the data structure separately,\footnote{We define $\deltahigh[\le i]$ and $\deltamid[\le i]$ similarly to $\delta_{\le i}$: $\deltahigh[\le i] \defeq \sum_{j=1}^{i} \deltahigh[j]$, and so is $\deltamid[\le i]$.}
  and then combine them into the output.
  \item During the $\rank(x)$ query, we sequentially examine the high part, mid part, and low part of the data structure. At each step, we narrow down the search interval for the desired index $i$, gradually approaching its exact value in the low part.
\end{itemize}
Below, we will provide the detailed construction of the three parts separately.

\paragraph*{High part.}
The high part sequence $\bk{\deltahigh[1], \ldots, \deltahigh[B^t]}$ is sparse and thus easy to store. As
\begin{align*}
    \sum_{i=1}^{B^t} \deltahigh = \sum_{i=1}^{B^t} \floor*{\frac{\delta_i}{2^{2h}}} \le \frac{1}{2^{2h}} \sum_{i=1}^{B^t} \delta_i = \frac{\Delta}{2^{2h}},
\end{align*}
there are at most $\Delta/2^{2h}$ non-zero entries in the high part. 
By setting a large parameter $h$, we can make the space usage of the high part very small such that we can regard the entire high part as redundant information.

Let $I \subseteq [1, B^t]$ be the set of indices of all non-zero entries in the high part, then $\abs{I} \le \Delta/2^{2h}$. The high part of the data structure consists of the following two components: 
\begin{itemize}
\item A predecessor data structure of the set $I$, where each $i \in I$ is associated with $\deltahigh[\le i]$. 
\item A predecessor data structure of the set $\BK{\delta_{\le i} : i\in I}$, where each $\delta_{\le i}$ is associated with $i$.\footnote{Note that all the $\delta_{\le i}$ are distinct for $i \in I$ as $\deltahigh \neq 0$ for all $i \in I$, hence $\BK{\delta_{\le i} : i\in I}$ is a valid set.}
\end{itemize}
According to \cref{lm:predecessor}, there is a compact implementation of predecessor data structures, with $O(\abs{I} \cdot w) = O(w \cdot \Delta / 2^{2h})$ bits of space and query time $O(\log \log U)$.

\paragraph*{Mid part.}
We will store the mid-part sequence $\bk{\deltamid[1], \ldots, \deltamid[B^t]}$ by an aB-tree with branching factor $B$ and size $B^t$, using \cref{lm:succincter}.\footnote{When the block size $m$ is smaller than $B^t$, we will use an incomplete aB-tree, that is why we need to support incomplete aB-trees in \cref{lm:succincter}.}
The details of the aB-tree are as follows.
\begin{itemize}
\item Each leaf of the aB-tree contains a single mid-part entry $\deltamid[i] \in [2^{2h}] \eqdef \Sigma$.
\item There is a label on each node that equals the sum of all the leaves $\deltamid$ in its subtree. As $\deltamid < 2^{2h}$ for each $i \in [1, B^t]$, we have $\Phi \subseteq \Bk{B^t \cdot 2^{2h}}$.
\item The premise of \cref{lm:succincter}, $B = O(\frac{w}{\log \bk{B^t + \abs{\Phi}}})$, can be satisfied by setting $h = t \log B$. Recalling that $B$ is a parameter with $B\log B = \frac{\eps\log U}{t}$, we have $B = \frac{\eps\log U}{t\log B} = O\bk{\frac{\log U}{\log \bk{B^t\cdot 2^{2h}}}} = O\bk{\frac{w}{\log \bk{B^t + \abs{\Phi}}}}$.
\item The size of the lookup table is
  \[O(B^{2t}\abs{\Sigma} + B^{3t}\abs{\Phi}^{2B}) = O(B^{3t}\abs{\Phi}^{2B}) \le O(B^{3t} \cdot (B^{3t})^{2B}) = O(B^{6tB + 3t}) \ll O(B^{10tB}) = O(U^{10\eps})\]
  words, which will be shared between blocks.
\item The aB-tree can support three types of queries within $O(t)$ time:
  \begin{enumerate}
  \item Given $i$, query $\deltamid[\le i]$.
  \item Given $v$, query the largest index $i_2$ with $\deltamid[\le i_2] \le v$.
  \item Given the value $v$, query the \emph{maximal interval} $[i_1, i_2]$ of the sequence $\bk{\deltamid[1], \ldots, \deltamid[B^t]}$ with respect to $v$, defined as follows.
  \end{enumerate}
  \begin{definition}
    The \defn{maximal interval} of a (non-negative) sequence $\bk{a_1, \ldots, a_{B^t}}$ with respect to value $v$ is defined as the interval $[i_1, i_2]$ formed by all indices $i$ with $a_{\le i} = v$. In the maximal interval, we have $a_{i_1 + 1} = a_{i_1 + 2} = \cdots = a_{i_2} = 0$, while $a_{i_1}, a_{i_2 + 1} > 0$, and $a_{\le i_1} = v$.
  \end{definition}
\end{itemize}

We now compute the space usage of the mid part. To use \cref{lm:succincter}, we need to first store the root label $\phi$ before storing the aB-tree, as \cref{lm:succincter} assumes free access to the root label.
After that, the aB-tree will occupy $\log \mathcal{N}\bk{B^t, \phi} + 2 $ bits where $\mathcal{N}\bk{B^t, \phi}$ is the number of sequences $\bk{\deltamid[1], \ldots, \deltamid[B^t]}$ with root label $\phi$. Recalling that $\phi$ is the sum of this sequence and
\[\phi = \sum_{i=1}^{B^t}\deltamid \le \sum_{i=1}^{B^t}\delta_i = \Delta,\]
the number of possible sequences $\bk{\deltamid[1], \ldots, \deltamid[B^t]}$ is at most $\binom{\Delta + B^t }{B^t - 1}$. Therefore, the number of bits used by the mid part (per block) is at most
\begin{align*}
  O(\log \Delta) + \log \mathcal{N}(B^t, \phi) + 2 = \log \binom{\Delta + B^t}{B^t - 1} + O(w).
\end{align*}

\paragraph*{Low part.}

In this basic data structure, the low part is directly stored in an array. Specifically, we store the sequence $\bk{\xlow[1], \ldots, \xlow[B^t]}$ one by one as $B^t$ integers each of $b$ bits in the memory.

Note that the low-part sequence has the following \defn{locally increasing} property which we will use in our query algorithms:
\begin{itemize}
  \item Restricted to a maximal interval $[i_1, i_2]$ of $\bk{\delta_1, \ldots, \delta_{B^t}}$, the subsequence $\bk{\xlow[i_1], \ldots, \xlow[i_2]}$ is \emph{strictly increasing}. This is because by the condition of \cref{prob:FID_within_block}, the sequence $\bk{\xlow[i_1] + 2^b \cdot \delta_{\le i_1}, \, \ldots \, , \, \xlow[i_2] + 2^b \cdot \delta_{\le i_2}}$ is strictly increasing, whereas $\delta_{\le i_1} = \cdots =\delta_{\le i_2}$ by the definition of  maximal interval.
  \end{itemize}

\paragraph*{Query algorithms.}
Now we can formally state our algorithms for the \partialsum and \rank queries in this intra-block data structure.

For the $\partialsum(i)$ query, we will query the high part, mid part, and low part of the data structure separately to get $\deltahigh[\le i]$, $\deltamid[\le i]$, and the $\xlow$:
\begin{itemize}
  \item To get $\deltahigh[\le i]$, we query $i$ on the first predecessor data structure of the high part, getting the largest index $i' \in I$ such that $i' \le i$, with its associated value $\deltahigh[\le i']$. As $i'$ is the last index before $i$ with a nonzero $\deltahigh[i']$, we have $\deltahigh[\le i] = \deltahigh[\le i']$, as desired. This takes $O(\log \log U)$ time.
  \item To get $\deltamid[\le i]$, we directly query it from the aB-tree of the mid part, which takes $O(t)$ time.
  \item To get $\xlow$, we directly read it out from the integer array of the low part, which takes $O(1)$ time.
\end{itemize}
Finally, the algorithm will return $\delta_{\le i} \cdot 2^b + \xlow =  \deltahigh[\le i] \cdot 2^{2h + b} + \deltamid[\le i] \cdot 2^{b} + \xlow$. The total time cost will be $O(t + \log \log U)$.

\smallskip

For the $\rank(x)$ query, we will sequentially query the high part, mid part, and low part of the data structure to narrow down the search interval of the desired index $i$:
\begin{itemize}
\item We first query the high part to locate $i$ within a maximal interval of $\bk{\deltahigh[1], \ldots, \deltahigh[B^t]}$. This is achieved by querying the second predecessor data structure of the high part, when we will get two adjacent indices $i_1, i_2 \in I$ with $\delta_{\le i_1} \le \floor{x/2^b} < \delta_{\le i_2}$, which locates the index $i$ within $[i_1, i_2)$. Moreover, as $i_1, i_2$ are adjacent indices in $I$, the interval $[i_1, i_2)$ forms a maximal interval of the sequence $\bk{\deltahigh[1], \ldots, \deltahigh[B^t]}$ with respect to the value $\vhigh \defeq \deltahigh[\le i_1]$. This takes $O(\log \log U)$ time.
\item Recall that the mid-part step needs to find the largest $i' \in [i_1, i_2)$ such that $\delta_{\le {i'}}$ does not exceed a threshold $\floor{x / 2^b}$. As all ${i'}$ share the same value $\deltahigh[\le {i'}]$, this is equivalent to finding the largest ${i'} \in [i_1, i_2)$ such that $\deltamid[\le {i'}]$ does not exceed $\floor{x / 2^b} - \vhigh \cdot 2^{2h}$. This can be done with one query to the aB-tree in the mid part.
  \begin{itemize}
  \item If the ${i'}$ we found satisfies the strict inequality $\delta_{\le {i'}} < \floor{x / 2^b}$, we conclude the query with $i = {i'}$, because no matter what the low parts are, we already know $\delta_{\le {i'}} \cdot 2^b + \xlow[i'] \le (\delta_{\le i'} + 1) \cdot 2^b - 1$ is strictly smaller than the queried key $x$, and $\delta_{\le {i' + 1}} \cdot 2^b + \xlow[i' + 1] \ge \delta_{\le {i' + 1}} \cdot 2^b$ is larger than $x$.
  \item If $\delta_{\le i'}$ is equal to the threshold $\floor{x / 2^b}$, we find the maximal interval $[i'_1, i'_2] \subseteq [i_1, i_2)$ of $(\delta_1, \ldots, \delta_{B^t})$ with $\delta_{\le i'_1} = \delta_{\le i'_2} = \floor{x / 2^b}$ by querying the aB-tree again. In this case, we can locate the answer $i$ to the query within the interval $[i'_1 - 1, \, i'_2]$, but the concrete answer depends on the information in the low part.
  \end{itemize}
  In both cases, the mid-part step takes $O(t)$ time.

\item Assuming the previous step encounters the second case, where an interval $[i'_1, \, i'_2]$ is known to have the same $\delta_{\le i'}$ for all $i'$, and the answer $i$ to the query is guaranteed to reside in $[i'_1 - 1, \, i'_2]$. Therefore, we need to find the largest $i \in [i'_1, i'_2]$ such that $\xlow[i] \le (x \bmod 2^b)$, and that will be the answer to the query. (If $\xlow[i'_1]$ is already larger than $(x \bmod 2^b)$, the answer to the query should be $i'_1 - 1$.) 
  Since $\bk[\big]{\xlow[i'_1], \ldots, \xlow[i'_2]}$ is strictly increasing, we can use binary search to find $i$ within $O(\log L)$ time where $L \defeq i'_2 - i'_1 + 1$; as $L \le B^t$, the running time of this step is bounded by $O(\log B^t) = O(t \log \log U)$.\footnote{Recall that $B \log B = \frac{\eps \log U}{t}$, we have $B \le \log U$, hence $t\log B \le t \log \log U$.}
\end{itemize}
Hence, the total time cost of the \rank query is at most $O(t\log \log U)$.

\paragraph*{Space usage of the block.} Recalling that in our construction, the high part, mid part, and low part of the data structure use $O(w \cdot\Delta / 2^{2h})$, $\log \binom{\Delta + B^t}{B^t - 1} + O(w)$ and $b \cdot B^t$ bits, respectively, the total space used by the intra-block data structure is
\begin{align*}
    \label{eq:space_block_simple}
    b \cdot B^t + \log \binom{\Delta + B^t}{B^t - 1} + O\bk*{w + w \cdot \Delta / 2^{2h}}.\numberthis
\end{align*}

\subsection{Performance of the Data Structure}

In this subsection, we will combine all parts of our construction to get the basic data structure, and check its time complexity, lookup table size, and redundancy to complete the proof of \cref{thm:simple_FID}.

\paragraph*{Time complexity.} Recall that our query algorithm will first query the inter-block data structure using $O(\log \log U)$ time to locate a block to query, and then query one intra-block data structure within $O(t \log \log U)$ time to get the answer. The time complexity of the whole data structure is $O(t\log \log U)$ per operation.

\paragraph*{Lookup table size.}
The only lookup table used by our data structure occurs in the mid part within each block, which is used for aB-trees due to \cref{lm:succincter}. As computed before, the size of this lookup table is $O(U^{10\eps})$ words, which is shared between all blocks, as desired.

\paragraph*{Redundancy.} 
The redundancy of the whole FID consists of the following two parts:
\begin{itemize}
\item The entire inter-block data structure is regarded as redundant. To store the endpoint sequence, we need to store $n/B^t + 1$ keys from $[U]$, which occupies $O(w \cdot n/B^t + w)$ bits. To store the predecessor data structure of the endpoint sequence, we also need at most $O(w\cdot n/B^t + w)$ bits according to \cref{lm:predecessor}.
  Recalling that $B \log B = \frac{\eps \log U}{t}$, we have $B \ge \sqrt{B \log B} = \sqrt{\frac{\eps \log U}{t}}$, and hence
  \begin{align*}
    O\bk*{w \cdot n/B^t} \le O\bk*{\frac{n \log U}{(\eps\log U / t)^{t/2}}}
    &\le O\bk*{\frac{n\log U}{(\log U/ t)^{t/4}}} \le O\bk*{\frac{n}{(\log U/ t)^{t/8}}},%
    \footnotemark \\
    &\le R \defeq \max\BK*{\frac{n}{(\log U / t)^{\Omega(t)}}, \, O(\log U)}. \numberthis \label{eq:less-than-R}
  \end{align*}
  \footnotetext{We need to check that $\log U \le \bk{\log U / t}^{t/8}$ for any $100 \le t \le \log U / \log \log U$ and sufficient large $U$. The derivative of the right-hand side equals, ignoring a positive multiplicative factor, $\log \log U - \log t - 1/\ln 2 \ge \log \log \log U - 1/ \ln 2  > 0$, which means $\bk{\log U / t}^{t/8}$ reaches its minimum at $t = 100$ where we have $\bk{\log U / t}^{t/8} \ge \bk{\log U / 100}^{100/8} > \log U$.}
  thus this part is covered by the desired amount of redundancy.
\item The redundancy caused by intra-block data structures, which we calculate below.
\end{itemize}
Recall that \eqref{eq:space_block_simple} upper bounds the space usage of each intra-block data structure. Taking a summation of \eqref{eq:space_block_simple} over all blocks, we get the total space consumption of the intra-block data structures:
\begin{align*}
    &\sum_{k=1}^{n/B^t} \bk*{b \cdot B^t + \log \binom{\Delta_k + B^t}{B^t - 1} + O\bk*{w + w \cdot \Delta_k / 2^{2h}}}\\
    {}={}& nb + \sum_{k=1}^{n/B^t} \log \binom{\Delta_k + B^t}{B^t - 1} + O\bk*{w + \frac{nw}{B^t} + \frac{w}{2^{2h}} \cdot \sum_{k=1}^{n/B^t} \Delta_k }\\
    {}\le{}& nb + \sum_{k=1}^{n/B^t} \log \binom{\Delta_k + B^t}{B^t - 1} + O\bk*{w + \frac{nw}{B^t} + \frac{nw}{2^{h}}}\\
    {}\le{}& nb + \sum_{k=1}^{n/B^t} \log \binom{\Delta_k + B^t}{B^t - 1} + O\bk*{R},\label{ineq:sum_of_block_total}\numberthis
\end{align*}
where the first inequality is due to \eqref{ineq:sum_of_delta_k}; the second inequality is because $h = t \log B$, thus $O(nw/2^{h}) = O(nw/B^t) \le O(R)$ according to \eqref{eq:less-than-R}.
We further have
\begin{align*}
    \label{ineq:sum_of_block_vandermonde}
    &\sum_{k=1}^{n/B^t} \log \binom{\Delta_k + B^t}{B^t - 1}
    {}={} \log \bk*{\prod_{k=1}^{n/B^t} \binom{\Delta_k + B^t}{B^t - 1}}
    {}\le{} \log \binom{\sum_{k=1}^{n/B^t} \Delta_k + n}{n - n/B^t}
    {}\le{} \log \binom{2^h n + n}{n}.\numberthis
\end{align*}
Comparing this quantity with the information-theoretic lower bound $\log \binom{U}{n}$ of FID, we get
\begin{align*}
    & \log \binom{U}{n} - \log \binom{2^h n + n}{n}
    {}={} \log \bk*{\frac{U}{2^h n + n} \cdot \frac{U - 1}{2^h n + n -1} \cdot \cdots \cdot \frac{U-n+1}{2^h n + 1}}\\
    {}\ge{}& n\log \bk*{\frac{U}{2^h n + n}}
    {}={} n\bk*{\log \frac{U}{2^h n} - \log \frac{2^h + 1}{2^h}}
    {}\ge{} nb - \frac{n}{2^h \ln 2},\label{ineq:compare_entropy}\numberthis
\end{align*}
where the last inequality is because $b \defeq \log \frac{U}{n} - h$ and $\log \bk{1 + 2^{-h}} \le 2^{-h} / \ln 2$. By plugging \eqref{ineq:sum_of_block_vandermonde} and \eqref{ineq:compare_entropy} into \eqref{ineq:sum_of_block_total}, the total space usage of all intra-block data structures is at most
\begin{align*}
    nb + \log \binom{2^h n + n}{n} + O\bk*{R}\le \log \binom{U}{n} + O\bk*{R + \frac{n}{2^h}} 
    = \log \binom{U}{n} + O(R),    
\end{align*}
where $O(n/2^{h}) = O(n/B^t)$ is covered by $R$ according to \eqref{eq:less-than-R}.

\smallskip

In conclusion, the total redundancy of our data structure for FIDs is bounded by $R$. Since we have constructed an FID with time complexity $O(t\log \log U)$, lookup table size $O(U^{10 \eps})$, and redundancy $R = \max \BK{n/(\log U/t)^{\Omega(t)}, O(\log U)}$, \cref{thm:simple_FID} follows.

\section{Advanced Data Structure for FIDs}
\label{sec:advanced}

In this section, we will prove \cref{thmt@@thmMainFID} by improving the basic data structure. 
\thmMainFID*

\begin{proof}
  We start by reviewing the bottleneck of \cref{thm:simple_FID}.
  In the basic data structure, the time bottleneck is that \rank queries take $O(t \log \log U)$ time in the low-part step. Specifically, suppose we are required to answer the query $\rank(x)$. The inter-block data structure, together with the high and mid parts of the intra-block data structure, uses $O(t + \log \log U)$ time to locate the desired index $i$ within a maximal interval $\Bk{i'_1, i'_2}$ of $\bk{\delta_1, \ldots, \delta_{B^t}}$.
Then, in the case where the low-part step is required, we need to further determine the maximum index $i \in [i'_1, i'_2]$ such that $\xlow \le \vlow \defeq (x \bmod 2^b)$. This can be formulated as a \rank query for the increasing sequence $\bk[\big]{\xlow[i'_1], \ldots, \xlow[i'_2]}$:
\begin{itemize}
\item $\rank(\vlow)$: Return the largest index $i \in [i'_1, i'_2]$, such that $\xlow \le \vlow$.
\end{itemize}
In the basic data structure, we directly store the sequence $\bk[\big]{\xlow[i'_1], \ldots, \xlow[i'_2]}$ as a sorted array using $bL$ bits and support \rank queries by binary search within $O\bk{\log L}$ time, where $L \defeq i'_2 - i'_1 + 1$ is the length of the subsequence. However, in the worst case, $L$ can be as large as the block size $B^t$, which makes the time of the binary search $O\bk{\log B^t} = O(t \log \log U)$. To address this bottleneck, we change the storage structure of the low part below.

\paragraph*{New construction for the low part.}
The key observation is that storing the subsequence $\bk[\big]{\xlow[i'_1], \ldots, \xlow[i'_2]}$ to support \rank and \partialsum queries can be viewed as managing a smaller FID.
Hence, we have the following \defn{unified data structure} to store it:
\begin{itemize}
\item If $L \le \Lthrd \defeq \log U$,
  then we still store the subsequence as a sorted array within $bL$ bits.
\item Otherwise, we will use the basic data structure in \cref{thm:simple_FID} (as a subroutine)
  to store the subsequence, with parameter $t' = O(1)$ and constant $\eps$ to be determined.%
  \footnote{\cref{thm:simple_FID} requires that the universe size $2^b$ of the basic data structure is at least a large polynomial of the number of keys $L \le B^t$. This condition holds because, on one hand, our choice of $B$ (i.e., $B \log B = (\eps \log U) / t$) implies $L \le B^t = U^{o(1)}$; on the other hand, $2^b = (U / n) / 2^h = (U / n) / B^t = \poly U$.}
  As $\BK[\big]{\xlow[i'_1], \ldots, \xlow[i'_2]} \subseteq [2^b]$ is a set of $L$ elements,
  the basic data structure will have query time $O(t' \log \log U) = O(\log \log U)$, and will use
  \begin{align*}
    \log \binom{2^b}{L} + \frac{L}{\bk{\log 2^b}^{\Omega \bk{1}}} + O(\log 2^b)
    \le L\log {\frac{e 2^b}{L}} + O(L + b)
    = Lb - L \log L + O(L + b)
    \le Lb
  \end{align*}
  bits of space, where the last inequality uses the fact $L > \Lthrd = \log U$ to ensure that $L\log L $ is asymptotically larger than $O(L)$ and $O(b) \defeq O(\log (U/n) - h) \le O(\log U)$.
  Hence, in this case, the basic data structure storing the subsequence also fits in $bL$ bits. We \emph{pad} 0's to the end of the encoding of the basic data structure until it occupies exactly $bL$ bits, ensuring memory alignment.
\end{itemize}
Clearly, this unified data structure allows us to answer $\rank$ and $\select$ queries of the subsequence $\bk[\big]{\xlow[1], \ldots, \xlow[B^t]}$ within $O(\log \log U)$ time:
\begin{itemize}
\item If $L \le \Lthrd$, then the subsequence is stored as a sorted array. For \rank queries, we can use binary search to get the desired index, which takes $O(\log L) \le O(\log \Lthrd) = O(\log \log U)$ time; for \select queries, we can directly read out the desired $\xlow$ within $O(1)$ time.
\item Otherwise, the subsequence is stored using the basic data structure in \cref{thm:simple_FID}, which can support \rank and \select queries within $O(\log \log U)$ time.
\end{itemize}

Further, we can store the entire low part using the above unified data structure for subsequences.
Specifically, we first divide the interval $[1, B^t]$ into all maximal intervals of $\bk{\delta_1, \ldots, \delta_{B^t}}$. For each maximal interval $[i'_1, i'_2]$ of length $L$, we store the corresponding subsequence of the low part $\bk[\big]{\xlow[i'_1], \ldots, \xlow[i'_2]}$ using the unified data structure in $bL$ bits. Finally, we concatenate all these unified data structures from the leftmost maximal interval to the rightmost one, obtaining a string of $b B^t$ bits, which serves as the encoding of the entire low part.

\smallskip

Finally, to obtain the advanced data structure, we start with the basic data structure and replace the encoding of the low part with the new construction above (the concatenation of unified data structures).
As our new construction for the low part occupies the same space ($bB^t$ bits) as before, the redundancy of this data structure is still $R = \max\BK{n/(\log U/t)^{\Omega(t)}, \, O(\log U)} = n/(\log U/t)^{\Omega(t)}$.
What makes the advanced data structure special is that the ``basic data structure'' appears twice here, once as the entire framework and once as the subroutines for maximal intervals in the low part. By embedding small instances of basic data structures within a larger framework of the same basic data structure in a non-recursive manner, we can improve the query time of FIDs without introducing any additional redundancy, as introduced below.

\paragraph*{Query algorithms.}

The query algorithms for our advanced data structure are similar to those of the basic data structure in \cref{sec:simple}: We first use the inter-block information to transform the original queries on the entire FID into \rank/\partialsum queries within each block. Then, a three-stage process involving the high, mid, and low parts will obtain the answer to the query. The only difference is that, before we access anything in the low part, we need to first compute the maximal interval of $(\delta_1, \ldots, \delta_{B^t})$ that we plan to access; by comparing the length of the interval with the threshold $\Lthrd$, we determine if the unified data structure for that maximal interval is stored as a sorted array or a basic data structure. The algorithms to answer \rank/\partialsum queries within each block are explained below.

\begin{itemize}
\item For the query $\rank(x)$, using the high and mid parts of the data structure, we can either answer the query directly or locate the desired index $i$ within a maximal interval $[i'_1, i'_2]$ of $\bk{\delta_1, \ldots, \delta_{B^t}}$. In the latter case, we read the encoding of the unified data structure for $\bk[\big]{\xlow[i'_1], \ldots, \xlow[i'_2]}$, which resides in the $\bk{(i'_1 -1)b+1}$-th bit to the $i'_2 b$-th bit in the encoding of the entire low part.
  Recall that comparing $i'_1 - i'_2 + 1$ with $\Lthrd$ will tell us whether the unified data structure is stored as a sorted array or a basic data structure for FIDs.
  After that, we perform the query $\rank(\vlow) = \rank(x \bmod 2^b)$ on the unified data structure to obtain the desired index $i$ within $O(\log \log U)$ time.
\item For the query $\partialsum(i)$, we first get $\delta_{\le i}$ by the mid and high parts. Then, instead of directly reading out $\xlow$ from the low part as before, now we also need to know the maximal interval $[i'_1, i'_2]$ of $\bk{\delta_1, \ldots, \delta_{B^t}}$ containing $i$ to help us access the low part. We use a similar algorithm to the \rank query, to compute this maximal interval $[i'_1, i'_2]$ with respect to $\delta_{\le i}$, and to extract the encoding of the subsequence $\bk[\big]{\xlow[i'_1], \ldots, \xlow[i'_2]}$. By querying $\select(i - i'_1 + 1)$ on the unified data structure of this subsequence, we get $\xlow$ in $O(\log \log U)$ time.
\end{itemize}
Both types of queries take $O(t + \log \log U)$ time, because the aB-trees in the mid part take $O(t)$ time, while other steps (including the high and low parts and the inter-block data structure) take $O(\log \log U)$ time per query. This meets the requirement in \cref{thm:main-fid}.

\paragraph*{Space of the lookup table.}
Finally, we check that the size of the lookup table introduced by \cref{thm:simple_FID} is also dominated by $R$. Recall that the lookup table consists of $O(U^{10\eps})$ words. As $U = n^{1+ \Theta(1)}$, we can assume there is a constant $\alpha > 1$ such that $U \le n^{\alpha}$. Then, we can set $\eps = 1/(20\alpha)$, which means that the number of \emph{bits} in the lookup table is $O(U^{1/(2 \alpha)} \log U) = O(n^{1/2} \log n)$, which is significantly smaller than $R = n/(\log U/t)^{\Omega(t)}$, as desired.

In summary, we get a data structure for static FID with query time $O(t + \log \log U)$ and redundancy $R = n/(\log U/t)^{\Omega(t)}$, which concludes the proof of \cref{thmt@@thmMainFID}.
\end{proof}

\section{Select and Partial Sum}
\label{sec:select}

In this section, we prove \cref{thm:select,thm:partialsum} by adjusting the \emph{basic data structure} introduced in \cref{sec:simple}. We will rely on the predecessor data structure from \cite{patrascu2006timespace} when the set to store is relatively dense:

\begin{lemma}[Similar to \cref{lm:predecessor}, see \cite{patrascu2006timespace}]
  \label{lm:pred-dense}
  For $U \le n \log^{t} n$, there is a predecessor data structure with associated values that uses $O\bk[\big]{n \log U + n \log V}$ bits of space and answers queries in $O(\log t)$ time.
\end{lemma}

We follow the same notations in \cref{sec:simple} in the following proofs.

\subsection{Select Dictionaries}

\thmSelect*

\begin{proof}
  Recall that we have divided the binary representations of keys into the \emph{high, mid, and low parts}, and in the high part, for each block of $B^t$ keys, we stored a predecessor data structure (\cref{lm:predecessor}) to compute $\deltahigh[\le i]$, which takes $O(\log \log n)$ time to answer each predecessor query. This is the only step exceeding $O(t)$ time in the process of answering \select queries.

  Instead of storing predecessor data structures for each block separately, here we store one large predecessor data structure for all $n$ keys with nonzero $\deltahigh$'s. It achieves the same functionality of computing $\deltahigh[\le i]$. The number of elements with nonzero $\deltahigh$'s is bounded by $n / 2^h = n / B^t$, thus the predecessor data structure stores $n / B^t \ge n / \log^t n$ elements from the range $[1, n]$. (If the number of nonzero $\deltahigh$'s is smaller than $n / B^t$, we add dummy elements until there are $n / B^t$ elements.) According to \cref{lm:pred-dense}, the predecessor data structure takes $O(\log t)$ time to answer each query, and takes $n \log n / B^t = n / (\log n / t)^{\Omega(t)}$ bits of space, which fits in our desired redundancy. Other parts of the data structure remain the same as in \cref{sec:simple}.

  When we perform a \select query, the above predecessor data structure in the high part will compute the prefix sum of the high part of the difference sequence, which takes $O(\log t)$ time per query; the aB-trees in the mid part takes $O(t)$ time to return the prefix sum of the mid part (within each block); finally, the low part reads $\xlow[i]$ directly to obtain the low part of the target key. The entire process takes $O(t)$ time.
\end{proof}

\subsection{Partial Sum on Integer Sequences}

\thmPartialSum*

\begin{proof}
    Let $x_i \defeq \sum_{j = 1}^i a_i$ be the partial-sum sequence of the input. The partial-sum problem is equivalent to storing a (multi-)set of keys $x_1 \le x_2 \le \cdots \le x_n$ supporting \select queries, i.e., a \emph{select dictionary}. The only distinction is that the difference $x_i - x_{i-1}$ between any two adjacent keys is bounded by $2^\l - 1$ in this problem. The data structure we design for partial-sum is similar to that of the select dictionaries, except that we adjust the parameters and change the number of bits in the \emph{high, mid, and low parts}:
    \begin{itemize}
    \item There is no high part.
    \item Let $B$ be a parameter such that $B \log B = \frac{\eps \log n}{t}$ for a small constant $\eps$, and let $h \defeq t \log B$. We call the $(\l - h)$ least significant bits of each $x_i$ the \emph{low part}, and store these bits directly using an array.
    \item The remaining bits $\floor{x_i / 2^{\l - h}}$ are called the \emph{mid part}. In their difference sequence $\delta_1, \ldots, \delta_n$ where $\delta_i \defeq \floor{x_i / 2^{\l - h}} - \floor{x_{i - 1} / 2^{\l - h}}$, each entry $\delta_i$ equals either $\floor{a_i / 2^{\l - h}}$ (i.e., the $h$ most significant bits of the input entry $a_i$) or $\floor{a_i / 2^{\l - h}} + 1$, and thus is in $[0, 2^h]$. Same as in \cref{sec:simple}, we divide $\delta_1, \ldots, \delta_n$ into blocks of size $B^t$ and use aB-trees to store them, supporting prefix-sum queries on $\delta_1, \ldots, \delta_n$.
    \end{itemize}

    Recall that $\mathcal{N}(B^t, \phi)$ represents the number of instances for an aB-tree with size $B^t$ and root label $\phi$ (i.e., the sum of entries in the aB-tree equals $\phi$), which is bounded by $(2^h + 1)^{B^t}$. The space usage of the mid part is thus
    \[
    \log \mathcal{N}(B^t, \phi) + 2 + O(\log n) \le B^t \log (2^h + 1) + O(\log n) \le B^t \cdot h + O(B^t / 2^h + \log n) = B^t \cdot h + O(\log n)
    \]
    bits per block, where $\log \mathcal{N}(B^t, \phi) + 2$ is the space usage of the aB-tree, and $O(\log n)$ is the space to store the root label $\phi$ of the aB-tree. Taking a summation of the space usage over all blocks, including the $n \cdot (\l - h)$ bits taken by the array in the low part, the inter-block information, and the lookup tables, we know the total space occupied by the data structure is at most
    \[
    n (\l - h) + \frac{n}{B^t} \cdot \bk*{
    B^t \cdot h + O(\log n)
    } + n^{0.1}
    \le n (\l - h) + n \cdot h + O(n \log n / B^t)
    \le n \l + n / (\log n / t)^{\Omega(t)}
    \]
    bits, as desired. Similar to \cref{thm:select}, each query takes $O(t)$ time.
\end{proof}

\section*{Acknowledgement}

The authors thank William Kuszmaul and Huacheng Yu for helpful suggestions on paper writing, and anonymous reviewers for pointing out important related works.

\bibliographystyle{alpha}
\bibliography{reference.bib}

\end{document}